\pdfoutput=1
\RequirePackage{ifpdf}
\ifpdf 
\documentclass[pdftex]{sigma}
\else
\documentclass{sigma}
\fi

\begin{document}
	
\allowdisplaybreaks

\renewcommand{\thefootnote}{$\star$}

\newcommand{\arXivNumber}{1509.07288}

\renewcommand{\PaperNumber}{094}

\FirstPageHeading
	
\ShortArticleName{Extended Hamiltonians, Coupling-Constant Metamorphosis and the Post--Winternitz System}
	
\ArticleName{Extended Hamiltonians, Coupling-Constant\\ Metamorphosis and the Post--Winternitz System\footnote{This paper is a~contribution to the Special Issue
on Analytical Mechanics and Dif\/ferential Geometry in honour of Sergio Benenti.
The full collection is available at \href{http://www.emis.de/journals/SIGMA/Benenti.html}{http://www.emis.de/journals/SIGMA/Benenti.html}}}
	
\Author{Claudia Maria {CHANU}, Luca {DEGIOVANNI} and Giovanni {RASTELLI}}
	
\AuthorNameForHeading{C.M.~Chanu, L.~Degiovanni and G.~Rastelli}
	
\Address{Dipartimento di Matematica, Universit\`a di Torino,  Torino, via Carlo Alberto~10, Italy}
\Email{\href{mailto:claudiamaria.chanu@unito.it}{claudiamaria.chanu@unito.it}, \href{mailto:luca.degiovanni@gmail.com}{luca.degiovanni@gmail.com}, \href{mailto:giovanni.rastelli@unito.it}{giovanni.rastelli@unito.it}}

\ArticleDates{Received September 26, 2015, in f\/inal form November 16, 2015; Published online November 24, 2015}
	
\Abstract{The coupling-constant metamorphosis is applied to modif\/ied extended Hamiltonians and suf\/f\/icient conditions are found in order that  the transformed high-degree f\/irst integral of the transformed Hamiltonian is determined by the same algorithm which computes the corresponding f\/irst integral of the original extended Hamiltonian. As examples, we consider the Post--Winternitz system and the 2D caged anisotropic oscillator.}
	
	
\Classification{37J35; 70H33}

\rightline{\it This article is dedicated to Sergio Benenti, our mentor, colleague and friend.}

\renewcommand{\thefootnote}{\arabic{footnote}}
\setcounter{footnote}{0}

\section{Introduction}
Classical and quantum Hamiltonian systems depending on a rational parameter $\lambda$ and admitting f\/irst integrals, or symmetry operators, of degree determined by $\lambda$ have been recently the object of research in integrable systems theory, with a particular interest on superintegrable and separable systems. In many examples the conf\/iguration manifolds of these systems are constant-curvature Riemannian or pseudo-Riemannian manifolds of  f\/inite dimension \cite{Bo,6bis,7bis,Ra2}, but some examples  of non-constant-curvature manifolds are also known \cite{KKMncf}.

In a series of articles, we proposed an algorithm for the construction, given a suitable $N$-dimensional Hamiltonian $L$, of classical $(N+2)$-dimensional Hamiltonians with f\/irst integrals depending on a rational parameter; this approach started from the analysis of the Jacobi--Calogero and Wolfes systems~\cite{CDR0}. The construction is not restricted to superintegrable or se\-pa\-rable systems, even if it allows to build new superintegrable systems from known ones~\cite{CDRsuext}.  Although it  involves a privileged coordinate system, the application of the algorithm is intrinsically characterized~\cite{CDRgen} and it is ultimately rooted into the geometry of the manifold (Poisson or symplectic) where $L$ is def\/ined, imposing conditions, as instance, on Riemannian curvature, on the geo\-metry of warped manifolds~\cite{CDRGhent} (see also, for example, \cite{MR, Ta}) and on a~particular type of master symmetries~\cite{CDRgen,Ra3}.

 The Hamiltonian systems admitting such an algorithmic construction are called ``extensions'' and  many of the known Hamiltonians with high-degree f\/irst integrals depending on $\lambda$ are indeed extensions of some other Hamiltonian~$L$~\cite{TTWcdr}. The algorithm consists essentially in the determination of the $\lambda$-dependent f\/irst integral through the power of a dif\/ferential operator generated by the Hamiltonian vector f\/ield of~$L$, applied to some suitable function~$G$. We stress the fact that our construction of the $\lambda$-dependent f\/irst integral does not assume integrability, separability or superintegrability of the system, dif\/ferently from all other approaches, and provides a compact expression of the real f\/irst integral itself.

With this article we modify  the algorithm in order to apply it to the Post--Winternitz (PW) system. The PW-system is a $\lambda$ dependent Hamiltonian system including a Kepler--Coulomb potential term~\cite{PW} originally obtained by coupling-constant metamorphosis (CCM)~\cite{H} of the Tremblay--Turbiner--Winternitz (TTW) system~\cite{TTW}. The CCM is a powerful tool for obtaining new integrable or superintegrable Hamiltonian systems from known ones (in particular, when applied to St\"ackel separable systems it is called ``St\"ackel transform''~\cite{St})  and has been recently extensively employed in the study and classif\/ication of superintegrable Hamiltonian systems (see~\cite{KMccm, 7bis} and references therein).

It is proved in \cite{TTWcdr} that the TTW system is what we call a~``modif\/ied extension'' and, therefore, that its $\lambda$-dependent f\/irst integral can be computed through the power of some operator~$W$ applied to a function $G$. In the following, we apply the CCM to modif\/ied extensions  and show in particular that the $\lambda$-dependent f\/irst integral of the PW system is equal to some  power of the CCM of the operator~$W$ applied to the same function $G$ appearing in the construction of the TTW system as extension. Therefore, even if the PW Hamiltonian cannot be written as a~modif\/ied extension in the same way of the TTW system, nevertheless the same algorithm for the determination of the f\/irst integral works as well. This suggests the  def\/inition of a class of Hamiltonian systems which includes all the systems we call  ``extended Hamiltonians'' and those systems related to them as the~PW is related to the~TTW.

In Section~\ref{section2} we recall the def\/inition of CCM and its connection with the PW system. In Section~\ref{section3} we review the theory of extended Hamiltonian systems.
In Section~\ref{section4}, where the main results are exposed, we study the application of CCM to extended systems and the results are completed by two examples: the TTW system, from which the PW system is obtained, and the caged anisotropic oscillator. In Section~\ref{section5} the results of the article are summarized,  a~new class of Hamiltonian systems is def\/ined, to include extended Hamiltonians and related systems; a new direction of research is suggested.

\section{Coupling-constant metamorphosis}\label{section2}

The coupling-constant metamorphosis (CCM) \cite{KMccm}  transforms  integrable or superintegrable systems in new integrable or superintegrable ones, by mapping f\/irst integrals in f\/irst integrals. We take the following statement from Theorem~1 of~\cite{PW} (where the CCM is called St\"ackel transform), Theorem~1 of~\cite{KMccm} and, ultimately, from~\cite{H}

 \begin{theorem}\label{theorem1}
Let us consider a Hamiltonian $H =\hat H - \tilde EU$ in canonical coordinates~$(q^i,p_i)$, where $\hat H(q^i,p_i)$ is independent
of the arbitrary parameter $\tilde E$ and $U(q^i)$, with an integral of the motion~$K$ $($depending on $\tilde E)$. If we define the CCM  of~$H$ and~$K$ as $\tilde H = U^{-1}( \hat H -E)$ and $\tilde K = K|_{\tilde E=\tilde{H}}$ then $\tilde K$ is an integral of the motion for~$\tilde H$.
 \end{theorem}

For example, in \cite{PW} the superintegrability of the Post--Winternitz (PW) system of Hamiltonian
\begin{gather}\label{PWe}
H_{\rm PW}=p_r^2+\frac {1}{r^2}\left(p_\phi^2+\frac 14 f_2\left(\frac \phi 2\right)\right)-\frac Q{2r},
\end{gather}
where
\begin{gather*}
f_2(x)= k^2 \left( \frac {\alpha}{\cos ^2 (kx)}+\frac {\beta}{\sin ^2 (kx)}\right),
\end{gather*}
is proved for $k\in \mathbb{Q}$ by writing it as result of the CCM   applied to the Tremblay--Turbiner--Winternitz (TTW) system \cite{TTWcdr,TTW}
\begin{gather}\label{TTWe}
H_{\rm TTW}=p_\rho^2+\frac{1}{\rho^2}\big(p_\theta^2+f_2(\theta)\big)-\tilde E \rho^2.
\end{gather}

Indeed,  by applying to it the CCM, the TTW system (\ref{TTWe}) becomes a Hamiltonian system of Hamiltonian
 \begin{gather*}
\tilde{H}=\frac1{\rho^2}\left(p_\rho^2+\frac{1}{\rho^2}\big(p_\theta^2+f_2(\theta )\big)-E\right),
 \end{gather*}
which coincides with (\ref{PWe}), through the coordinate change $\rho=\sqrt{2r}$, $\phi=2\theta$, and by setting $E=Q/2$.

 In this example a system on the Euclidean plane is mapped into another system on the same manifold. This is not always the case: by applying the CCM to the 2D caged oscillator \cite{KKM, VE}
 \begin{gather}\label{CO}
  H_{\rm co}=\frac 12 p_y^2+\frac 12 p_x^2+\omega^2\big(k^2x^2+y^2\big)+\frac b{x^2}+\frac{c}{y^2},
  \end{gather}
  with $\tilde E=-c$, we get
 \begin{gather}\label{COst}
 \tilde{H}= y^2\left(\frac 12 p_y^2+\frac 12 p_x^2+\omega^2\big(k^2x^2+y^2\big)+\frac{b}{x^2}-E\right),
 \end{gather}
 which is a system on the Poincar\'e half-plane.

\section{Extensions}\label{section3}

 In \cite{TTWcdr} we show that the TTW system can be written as a {\it modified extension}. We recall that a~$(N+2)$-dimensional Hamiltonian~$H$  is a modif\/ied extension of the $N$-dimensional Hamilto\-nian~$L$  if
 \begin{enumerate}\itemsep=0pt
  \item there exist canonical coordinates $(u,p_u)$ such that
 \begin{gather}\label{Hext}
 H=\frac 12 p_u^2+f(u)+\left(\frac {m}n\right)^2\alpha(u) L, \qquad m,n\in \mathbb{N}{\setminus}\{0\},
  \end{gather}
and the Hamiltonian $L$ does not depend on $(u,p_u)$;
\item
for some constants $c$ and $L_0$ not both vanishing, the equation
\begin{gather}
X_L^2(G)=-2(cL+L_0)G,\label{eqG}
\end{gather}
 where $X_L$ is the Hamiltonian vector f\/ield of~$L$, admits a solution~$G$, such that~$X_L(G)\neq 0$;
\item
the functions $\alpha$ and $f$ are those given in Table~\ref{tab1}.
 \end{enumerate}
Then, given the operator  $W$ def\/ined by
 \begin{gather*}
 W(G_\nu)=\left(p_u+\frac{\mu}{\nu^2}\gamma(u) X_L\right)^2(G_\nu)+\delta(u)G_\nu,
 \end{gather*}
 with $(\mu,\nu)=(m,n)$ if $m$ is even, $(\mu,\nu)=(2m,2n)$ if $m$ is odd, $\gamma$ and $\delta$ def\/ined as in Table~\ref{tab1}
 and $G_\nu$ being the $\nu$-th element  of
the recursion
\begin{gather}\label{eqGn}
 G_1=G, \qquad G_{\nu+1}=X_L(G_1)G_\nu+\frac 1\nu G_1X_L(G_\nu),
\end{gather}
we have  that
\begin{gather*}
K=W^{\frac \mu 2}(G_\nu)
\end{gather*}
is a f\/irst integral of~$H$.

In Table~\ref{tab1},
 $A$ and $\kappa$ are arbitrary constants and
 the functions $S_\kappa$ and $T_\kappa$ are the trigonometric tagged functions
 \begin{gather*}
 S_\kappa(x)= \begin{cases}
 \dfrac{\sin\sqrt{\kappa}x}{\sqrt{\kappa}}, & \kappa>0, \\
 x, & \kappa=0, \\
 \dfrac{\sinh\sqrt{|\kappa|}x}{\sqrt{|\kappa|}}, & \kappa<0,
 \end{cases}
 \qquad
 C_\kappa(x)= \begin{cases}
 \cos\sqrt{\kappa}x, & \kappa>0, \\
 1, & \kappa=0, \\
 \cosh\sqrt{|\kappa|}x, & \kappa<0,
 \end{cases}
 \\
 T_\kappa(x)=\frac {S_\kappa(x)}{C_\kappa(x)}
 \end{gather*}
 (see \cite{CDRraz} and \cite{7ter} for a summary of their properties).

 \begin{table}[t]
 \centering \caption{Functions involved in the modif\/ied-extensions of $L$.}\vspace{1mm}\label{tab1}
 \begin{tabular}{|c|c|c|}
 \hline
 &$ c=0$ & $ c \neq 0$ \\
  \hline
  $ \alpha = -{\gamma}'= $&$A$ &  $\dfrac {c}{S_\kappa^2( c u)}$ \tsep{5pt}\\
  $f = \dfrac{m^2}{n^2}L_0\gamma^2 +\dfrac{f_0}{\gamma^2}= $ & $  \dfrac{m^2}{n^2}L_0A^2u^2 +\dfrac{f_0}{A^2u^2}$ &
  $\dfrac {m^2}{n^2}\dfrac {L_0}{T_\kappa^2(c u)}+f_0 T_{\kappa}^2(c u)$  \\
 $ \gamma = $& $-Au$ & $\dfrac 1{T_\kappa( c u)}$\cr
 $\delta = \dfrac{2f_0}{\gamma^2} = $ & $\dfrac{2f_0}{A^2u^2} $ & $2f_0T^2_\kappa( c u) $  \bsep{6pt}\\
 \hline
 \end{tabular}
 \end{table}

We remark that
 \begin{itemize}\itemsep=0pt
 \item
 If (\ref{eqG}) has a solution  for $c\neq 0$, then  we may assume without loss of generality $L_0=0$, because $L$ is determined up to additive constants.

\item
 In the case of natural Hamiltonians $L$, the possibility of f\/inding extensions of $L$ is strictly related to the geometry of the base manifold of $L$ \cite{CDRgen}. An intrinsic characterization of extended Hamiltonians $H$ is given in~\cite{CDRgen}.

 \item
  The extension $H$ of a   Hamiltonian $L$ with $k$ functionally independent f\/irst integrals is a Hamiltonian with $k+2$ functionally independent f\/irst integrals~\cite{TTWcdr}. Indeed,  $H$ itself and $W^{\frac \mu 2}(G_\nu)$ are  functionally independent f\/irst integrals of $H$ together with $L$ and all its functionally independent f\/irst integrals. If~$L$ is maximally superintegrable, then also $H$ is.

 \item The dynamical equations in $(u,p_u)$ are always separated from those in the variables $(q^i,p_i)$ of $L$, being $L$ a constant of motion for $H$.
  \end{itemize}

 \section{Coupling-constant metamorphosis of extended Hamiltonians}\label{section4}

 By Theorem~\ref{theorem1}, it follows that the CCM can be applied to any modif\/ied extension (\ref{Hext}) by setting $\tilde E=-f_0$, $U=1/\gamma^2$. Therefore, in this case we have
 \begin{gather}\label{Hhat}
 \hat H=\frac 12 p_u^2-\left(\frac {m}n\right)^2\gamma' L +\dfrac{m^2}{n^2}L_0\gamma^2 ,
\\
\label{WE}
  W=\left(p_u+\frac{\mu}{\nu^2}\gamma(u) X_L\right)^2-2\frac{\tilde{E}}{\gamma^2} .
  \end{gather}
Moreover, the function
  \begin{gather}\label{Htil}
  \tilde H = \gamma^2\big(\hat H- E\big)=
  \frac {\gamma^2}2 p_u^2-\left(\frac {m}n\right)^2\gamma^2\gamma'(u) L +\dfrac{m^2}{n^2}L_0\gamma^4-E\gamma^2
  \end{gather}
  is the transformed Hamiltonian and
  \begin{gather*}
  \tilde K=\big(W^{\frac \mu 2}(G_\nu)\big)_{|\tilde E=\tilde H},
  \end{gather*}
  the transformed f\/irst integral of $\tilde H$.

A natural question is if $\tilde K$ is again given by a power of some operator applied to some function.
  The main result of this paper is that the answer is positive: $\tilde K$ can be computed  by applying~$\mu/2$ times
the operator
 \begin{gather}\label{tilW}
  \tilde W=(W)|_{\tilde E=\tilde H}=\left(p_u+\frac{\mu}{\nu^2}\gamma X_L\right)^2+2\big(E-\hat H\big),
 \end{gather}
 to the same function $G_\nu$ used for the determination of $K$. Indeed,
  \begin{proposition} \label{P2}
  The transformed first integral $\tilde K$ of \eqref{Htil}  is $\tilde W^{\frac \mu 2}(G_\nu)$,
  where
  \begin{gather}\label{tilWe}
  \tilde W=2\left( \frac{\mu}{\nu^2}\gamma p_uX_L -\frac{\mu^2}{\nu^2}\big(\big(c\gamma^2-\gamma'\big)L+\gamma^2 L_0\big) +E\right)
  \end{gather}
  and $G_\nu$ is computed via the recursion~\eqref{eqGn}.
 \end{proposition}

 \begin{proof}
We consider the iterated application of $\tilde W$. Being  $W$ and $\tilde{H}$ given by~(\ref{WE})
 and~(\ref{Htil}) respectively, we have
 $W(\tilde H)=\tilde H W$ because $X_L(\tilde{H})=0$. Moreover, $\tilde W X_L=X_L\tilde W$, therefore
  \begin{gather}\label{me}
  \big(W^{\frac \mu 2}\big)|_{\tilde E=\tilde H}= \big(W|_{\tilde E=\tilde H}\big)^{\frac \mu 2}=\tilde W^{\frac \mu 2}.
  \end{gather}
  Finally, we remark that $G_\nu$ does not depend on $\tilde{E}$, as well as~$L$.
  The explicit form of $\tilde{W}$ follows by expanding~(\ref{tilW}), inserting~(\ref{Hhat}) in it and by applying the formula~\cite{CDRraz}
  \begin{gather*}
  X_L^2 (G_\nu)=-2\nu^2(cL+L_0)G_\nu.\tag*{\qed}
  \end{gather*}
  \renewcommand{\qed}{}
 \end{proof}

\begin{remark} \label{remark1} Let us consider the  CCM of a natural $2N$-dimensional Hamiltonian with a generic potential  $U(q^1,\ldots,q^N,u)$. If $G_\nu$ and $L$ are independent from~$\tilde E$, then the condition for ha\-ving~(\ref{me}) is, from the proof of Pro\-po\-si\-tion~\ref{P2},
 \begin{gather*}
 X_L(\tilde H)=-\frac 1{U^2}\big(\hat H-E\big)X_L(U)=0,
 \end{gather*}
 that, by requiring its validity on the whole space, is equivalent to
 \begin{gather*}
 X_L(U)=0.
 \end{gather*}
 For $L(q^i,p_i)$  such that  $\partial L/\partial p_i\neq0$, $i=1,\ldots,N$, the condition of above is equivalent to~$U(u)$ and we go back to the result of Proposition~\ref{P2}. Other situations  are possible for dif\/ferent~$L$. We do not consider here the reduction of the system to the submanifold
 $\hat H=E$, i.e., $\tilde H=0$, and its f\/irst integrals.
 \end{remark}

 \begin{remark}\label{remark2}
 Up to a rescaling of $\tilde{u}=\tilde{u}(u)$ such that
 \begin{gather*}
 \frac{d\tilde{u}}{du}=\frac{1}{\gamma(u)}\qquad  \mbox{and} \qquad p_{\tilde{u}}=\gamma p_u,
 \end{gather*}
 the Hamiltonian (\ref{Htil}) is of the form (\ref{Hext}),
 and the operator $\tilde{W}$ def\/ined by~(\ref{tilWe})  becomes
 \begin{gather*}
 \tilde{W}= \frac{2\mu}{\nu^2}p_{\tilde{u}}X_L +2\frac{\mu^2}{\nu^2}\delta_1(\tilde{u})L+ \delta_2(\tilde{u}),
 \end{gather*}
 where the functions $\delta_i$ are given in Table~\ref{tabdelta}.
The general (i.e., independent of~CCM) conditions allowing the existence of
 f\/irst integrals generated by such type of operator will be analysed elsewhere.
 \end{remark}

 \begin{table}[t]
\centering
 \caption{Functions $\delta_1$ and $\delta_2$.}\vspace{1mm}\label{tabdelta}
  		\begin{tabular}{|c|c|c|c|}
 			\hline
 		&$c\neq0,$ $\kappa\neq 0$, 	&$c\neq 0$, $\kappa=0$,  & $c=0$, $L_0\neq 0$\\
 			\hline
 		$\delta_1$ &	$
 			 \dfrac{c\kappa\big(1+C_\kappa(cu)^2\big)}{1-C_\kappa(cu)^2}= \dfrac{c\kappa}{\tanh(\kappa c\tilde{u})}
 			$ &$
 			 \dfrac{2}{cu^2}=\dfrac{1}{\tilde{u}}$ &  $A$ \tsep{10pt} \\
 		$\delta_2$ &	$2E+\dfrac{L_0\mu^2\kappa}{\nu^2}\left(\dfrac{1}{\tanh(c\kappa\tilde{u})}-1\right) $ &
 		$  2E+\dfrac{L_0\mu^2\kappa}{\nu^2c\tilde{u}}$ &
 			$  2E+2\frac{\mu^2}{\nu^2}\dfrac{L_0A^2}{e^{2A\tilde{u}}}$\tsep{10pt}\bsep{8pt}  \\ 		
 			\hline
 		\end{tabular}
 \end{table}

 \subsection{Example 1: the TTW system}
 In \cite{TTWcdr} it is shown that the TTW system~(\ref{TTWe}) is a modif\/ied extension.
 Indeed,
 the extension of the Hamiltonian
 \begin{gather}\label{LTTW}
  L=\frac 12 p_\psi^2+\frac{c_1+c_2\cos \psi}{\sin^2\psi},
 \end{gather}
 (satisfying (\ref{eqG}) for $c=1$, $L_0=0$ and $G=p_\psi \sin\psi$)
 for $\kappa=0$, that is for
   $\gamma=1/u$,
  is
 \begin{gather}\label{TTWex0}
  H=\frac 12 p_u^2+\frac{m^2}{n^2u^2}\left(\frac 12p_\psi^2+\frac{c_1+c_2\cos \psi}{\sin^2\psi}\right)+f_0 u^2.
 \end{gather}
 The rescaling $u=\rho$, $\psi=2k \theta$,  the change of parameters  \begin{gather} \label{param}
 \frac{m}{n}=k,\qquad c_1=\alpha +\beta, \qquad c_2=\beta-\alpha,
 \end{gather} and the position $\tilde{E}=-2f_0$, transform~(\ref{TTWex0}) into the Hamiltonian $H_{\rm TTW}$ of~(\ref{TTWe}) multiplied by~2.

 The PW Hamiltonian (\ref{PWe}), instead, is not a modif\/ied extension, because the Kepler--Coulomb term $\frac Q{2r}$ cannot be included  in the form of~$f$ given in Table~\ref{tab1}.

 By applying the CCM based on  $\tilde E=-f_0$ as in Theorem~\ref{theorem1}  to the Hamiltonian~(\ref{TTWex0}), we get
 \begin{gather}\label{HPWu}
 \tilde{H}= \frac 1{2u^2} p_u^2+\frac{m^2}{n^2u^4}\left(\frac 12p_\psi^2+\frac{c_1+c_2\cos \psi}{\sin^2\psi}\right)-\frac{E}{u^2},
 \end{gather}
 and, by performing the rescaling $u^2=2r$,  we obtain
\begin{gather}\label{HPWr}
\tilde H=  \frac{1}{2}p_r^2+\frac{m^2}{4n^2r^2}\left( \frac{1}{2}p_\psi^2+\frac{c_1+c_2\cos \psi}{\sin^2\psi}\right)-\frac E{2r},
\end{gather}
which is, after (\ref{param}), and the rescaling $\psi=2k\phi$ together with the change of parameters~$Q=2E$, one half of the Hamiltonian~(\ref{PWe}).

Then, the operator generating f\/irst integrals of (\ref{HPWu}) for any rational~$m/n$ is
\begin{gather*}
\tilde{W}=2\left(\frac{\mu}{\nu^2u}p_u X_L-\frac{\mu^2}{\nu^2u^2}L + E\right),
\end{gather*}
with $(\mu,\nu)=(m,n)$ for $m$ even, $(\mu,\nu)=(2m,2n)$ for $m$ odd and where $X_L$ is the Hamiltonian vector f\/ield of (\ref{LTTW}). The function $G_\nu$ is recursively determined by
\begin{gather*}
 G_1=(\sin \psi) p_\psi, \qquad G_{\nu+1}=X_L(G_1)G_\nu+\frac 1\nu G_1X_L(G_\nu).
 \end{gather*}
By the rescaling $u^2=2r$, we get the operator generating f\/irst integrals for~(\ref{HPWr}):
\begin{gather*}
\tilde{W}=\frac{2\mu}{\nu^2}p_r X_L-\frac{\mu^2}{\nu^2r}L + 2E.
\end{gather*}

\subsection{Example 2: the caged anisotropic oscillator}

 In order to write the Hamiltonian on the Poincar\'e half-plane (\ref{COst})  as the CCM of a modif\/ied extension, we
need to express the  caged oscillator Hamiltonian (\ref{CO}) as a modif\/ied extension.
 From Section~\ref{section3}, we know that  the expression of a modif\/ied extension in a  plane  when  $c=0$ is
 \begin{gather}\label{CO2}
H_{m,n}= \frac 12p_u^2-\frac{m^2}{n^2}\gamma'L+\frac {m^2}{n^2}L_0\gamma^2+\frac {f_0}{\gamma^2},
 \end{gather}
 where $\gamma=-Au$. From \cite{TTWcdr} we know that  the most general one-dimensional natural Hamiltonian $L(p_q,q)$ admitting an extension for $c=0$, with $G=(a_1q+a_2)p_q$, is
 \begin{gather*}
 L=\frac 12 p_q^2+\frac {L_0}{4a_1^2}(a_1q+a_2)^2+\frac{c_1}{(a_1q+a_2)^2}+c_2,
 \end{gather*}
 being $A$, $a_i$, $c_i$, $L_0$ real constants.
 By extending  the Hamiltonian  $L$ into (\ref{CO2}), we can  write the CCM of $H_{m,n}$ according to (\ref{Htil}), obtaining
  \begin{gather}\label{PWe2}
  \tilde H_{m,n}= \gamma^2\left(\frac 12 p_u^2-\frac{m^2}{n^2}\gamma'L+\frac {m^2}{n^2}L_0\gamma^2-E'\right).
  \end{gather}
By comparing (\ref{COst}) with  (\ref{PWe2}) we obtain that the kinetic terms coincide for
\begin{gather*}
u=y,  \qquad \gamma=-y, \qquad q=\frac mn x+x_0.
\end{gather*}
The choice of $x_0=-\frac{a_2}{a_1}$ allows to write $a_1q+a_2=a_1\frac mn x$ and, consequently, we have the identif\/ications
\begin{gather*}
\frac {m^2}{n^2}=4k^2, \qquad L_0=\frac{\omega^2}{4k^2}, \qquad c_1=a_1^2b, \qquad E'=E+4c_2k^2.
\end{gather*}
 Therefore, the f\/irst integral is
 \begin{gather*}
 K=\tilde W ^{\frac \mu 2}G_\nu,
 \end{gather*}
 where
 \begin{gather*}
 \tilde W=-\frac{2\mu}{\nu^2}yp_yX_L-8k^2- 4\omega^2 y^2+2E+8c_2k^2,
 \\
 L=\frac 1{8k^2}p_x^2+\frac{\omega^2}{4}x^2+\frac b{4k^2x^2}+c_2,
 \end{gather*}
 with  $(\mu,\nu)=(m,n)$ for $m$ even, $(\mu,\nu)=(2m,2n)$ for~$m$ odd, and~$G_\nu$ given by the recursion~\cite{TTWcdr}
 \begin{gather*}
  G_1=a_1xp_x, \qquad G_{n+1}=X_L(G_1)G_n+\frac 1n G_1X_L(G_n).
  \end{gather*}
We remark that other choices of rescaling and changes of parameters are possible, leading to dif\/ferent (but essentially equivalent) $L$ and $\tilde{W}$.

 \section{Conclusions}\label{section5}

 In this article we proved that,
 for any modif\/ied extension, there exists a specif\/ic CCM of it which maintains the most distinctive property of an extension: the determination of a f\/irst integral via powers of an operator applied to a suitable function. This fact suggests the def\/inition of a~new class  of Hamiltonian systems  including the Post--Winternitz system as well as all extended Hamiltonian systems.

 In \cite{CDRGhent} we introduced the idea of warped product of Hamiltonian systems. Given two symplectic, or Poisson, manifolds $M$ and $N$ with Hamiltonians $H_M$ and $H_N$, we consider on $M\times N$, endowed with the product of the symplectic or Poisson structures of $M$ and $N$, the Hamiltonian
 \begin{gather*}
 H=\alpha H_M +\beta H_N,
 \end{gather*}
 where $\alpha$ and $\beta$ are functions on $M\times N$ and we call it the ``warped product'' of $H_M$ and $H_N$. If $H_M$, $H_N$ are natural Hamiltonians, $M$, $N$ are cotangent bundles with Riemannian man\-ifolds~$B_M$, $B_N$ as base manifolds and~$\alpha$, $\beta$ are functions of the product $B_M\times B_N$, then the metric tensor of $H$ is the standard warped product of the metrics on~$B_M$,~$B_N$.

 All extended Hamiltonian systems, together with their CCM considered in this paper, are clearly the warped product of two Hamiltonians: one depending on $(u,p_u)$ solely, the other being~$L$. Indeed, the symplectic structure we are using on~$H$ is simply the product  of the lower-dimensional canonical symplectic structures. Therefore, we may imagine a class of Hamiltonian systems of ``warped-power" type determined as follows
 \begin{itemize}\itemsep=0pt
 	\item their Hamiltonian $H$ is the warped product of a f\/inite number of other Hamiltonians,
 	
 	\item $H$ admits a constant of the motion determined by the power of an operator applied to some suitable function def\/ined on the same domain of~$H$.
 	\end{itemize}
Such a class includes all the systems we call ``extensions of Hamiltonian systems'', together with their CCM as described in this paper, and  is naturally parametrized by a natural number at least: the power of the operator generating the f\/irst integral.

Finally, Remark~\ref{remark2} suggests a new direction of research, by proposing an alternative form of the operator involved in the extension procedure.

\pdfbookmark[1]{References}{ref}
\LastPageEnding

\end{document}